\documentclass[sigconf,screen,nonacm]{acmart}

\AtBeginDocument{%
  }

\copyrightyear{2026}
\acmYear{2026}
\setcopyright{cc}
\setcctype{by-nc-nd}
\acmConference[ApPLIED'26]{Advanced Tools, Programming Languages, and PLatforms for Implementing and Evaluating algorithms for Distributed systems }{July 06--10, 2026}{Egham, United Kingdom}
\acmBooktitle{Advanced Tools, Programming Languages, and PLatforms for Implementing and Evaluating algorithms for Distributed systems (ApPLIED'26), July 06--10, 2026, Egham, United Kingdom}
\acmDOI{10.1145/3820355.3820388}
\acmISBN{979-8-4007-2462-6/2026/07}

\usepackage[vlined,ruled,linesnumbered,noresetcount]{algorithm2e}
\usepackage[noend]{algpseudocode}
\usepackage{newfloat}
\usepackage[inline]{enumitem}
\usepackage{color,soul}
\usepackage{multicol}
\usepackage{tabularx, booktabs}
\usepackage{cleveref}

\newtheorem{theorem}{Theorem}[section]
\newtheorem{lemma}{Lemma}[section]

\usepackage[ruled,vlined]{algorithm2e}
\usepackage{tikz}
\usetikzlibrary{arrows.meta}
\tikzset{
    timeline/.style={thin, gray!50},
    timeline dashed/.style={thin, gray!50, dashed},
    msg/.style={-{Stealth[length=4pt, width=3pt]}, semithick},
    msg lost/.style={-{Stealth[length=4pt, width=3pt]}, semithick, dashed, opacity=0.4},
    ev/.style={circle, fill, inner sep=1.2pt},
    plabel/.style={font=\small\bfseries, anchor=east},
    mlabel/.style={font=\scriptsize},
}

\SetCommentSty{mycommfont}

\makeatletter
\newcommand{\removelatexerror}{\let\@latex@error\@gobble}
\makeatother

\SetAlgoCaptionSeparator{.}
\DeclareFloatingEnvironment[placement={!ht}]{myalgo}
\DontPrintSemicolon

\newcommand{\bg}[1]{\left< #1 \right>}

\newcommand{\bc}[1]{\left\{ #1 \right\}}
\newcommand{\goesto}{\rightarrow}

\newcommand{\remove}[1]{}

\newcommand{\removeinfinal}[1]{#1}

\SetKwFunction{RelBlockOp}{releaseBlock}

\SetKwFunction{LE}{LeaderElect}
\SetKwFunction{NDecide}{NameDecide}

\SetKwFunction{LeaderElect}{LeaderElect}
\SetKwFunction{OpRead}{Read} \SetKwFunction{OpWrite}{Write}
\SetKwFunction{TAS}{test-and-set} \SetKwFunction{Reset}{reset}
\SetKwFunction{FAI}{fetch-and-increment}
\SetKwFunction{CAS}{CAS}
\SetKwFunction{CFG}{compare\_and\_fg}

\SetKwData{OK}{OK}
\SetKw{Wait}{await} \SetKwData{PID}{PID} \SetKwData{In}{in}
\SetKwData{win}{win} \SetKwData{lose}{lose}
\SetKwData{True}{true} \SetKwData{False}{false}
\SetKwFunction{Signal}{signal} \SetKwFunction{WaitSignal}{wait}
\SetKwFunction{WaitAny}{wait-any}
\SetKwFunction{InitOp}{init}

\SetKwFunction{AlgoBrack}{()}
\SetKwFunction{AlgoBrackL}{(}
\SetKwFunction{AlgoBrackR}{)}

\SetKwInput{KwPrivVar}{Private variables}
\SetKwInput{KwSharVar}{Shared variables}
\SetKwInput{KwHelpFun}{Helper functions}
\SetKwInput{KwPrecon}{Precondition}

\newcommand{\States}{\ensuremath{\mathcal{S}}}
\newcommand{\sinit}{\ensuremath{s_{init}}}
\newcommand{\Ops}{\ensuremath{\mathcal{O}}}
\newcommand{\Resps}{\ensuremath{\mathcal{R}}}
\newcommand{\Procs}{\ensuremath{\mathcal{P}}}

\newcommand{\Typ}{\ensuremath{\tau}}

\SetKwFor{ForEachP}{foreach}{do in parallel}{end parallel}
\SetKwRepeat{RepeatF}{loop forever}{end}

\SetKwFunction{NewBlockOp}{AllocNewBlock}
\SetKwFunction{SetBlockOp}{chngCurBlock}
\SetKwFunction{GetBlockOp}{getCurBlock}
\SetKwFunction{LL}{LL}
\SetKwFunction{SC}{SC}
\SetKwFunction{ECAS}{ECAS}
\SetKwFunction{OpLock}{Enter}
\SetKwFunction{OpUnlock}{Exit}

\SetKwFunction{HelpBegin}{HelperBegin}
\SetKwFunction{HelpEnd}{HelperEnd}
\SetKwFunction{HelpSucc}{HelperCC}
\SetKwFunction{HelpSuccSubA}{TryToReuseBlock}
\SetKwFunction{HelpSuccSubB}{DoneReusingBlock}

\SetKwFunction{Pent}{Pseudo-Enter}
\SetKwFunction{Pext}{Pseudo-Exit}

\removeinfinal{
\pagenumbering{arabic}
}

\begin{document}

\title{Analyzing Linearizability in Relativistic Distributed Systems}

\author{Kahbod Aeini}
\email{kaeini@uwaterloo.ca}
\orcid{0009-0004-9161-4792}
\affiliation{
  \institution{University of Waterloo}
  \department{Electrical and Computer Engineering}
  \city{Waterloo}
  \country{Canada}
}
\author{Wojciech Golab}
\email{wgolab@uwaterloo.ca}
\orcid{0000-0002-8891-256X}
\affiliation{
  \institution{University of Waterloo}
  \department{Electrical and Computer Engineering}
  \city{Waterloo}
  \country{Canada}
}

\renewcommand{\shortauthors}{Aeini and Golab}

\begin{abstract}
Einstein's theory of relativity correctly predicted that time is relative, and subject to both kinematic and gravitational dilation.
Therefore, executions of distributed systems cannot always be modeled as sequences of events totally ordered according to wall clock time.
To address this fundamental problem, Gilbert and Golab formulated a generalization of Herlihy and Wing's linearizability property for shared objects,
	which they called \emph{relativistic linearizability}, and introduced a collection of theoretical tools to facilitate rigorous analysis.
While they conjectured that several widely-studied classically linearizable algorithms are also relativistically linearizable, their work
	stopped short of presenting formal proofs of correctness, as pointed out recently by Jayanti.
In this paper, we explain how Gilbert and Golab's techniques can be used to establish relativistic linearizability for
	a replicated state machine, as well as variations of the
	widely studied read/write register construction of Attiya, Bar-Noy and Dolev (ABD).
Our results establish a stronger form of relativistic linearizability than Jayanti's central theorem for these asynchronous algorithms.
\end{abstract}




\maketitle

\section{Introduction} \label{sec:intro}

As humanity continues to expand scientific activity beyond Earth, distributed computing will play an increasingly important role in enabling reliable communication, coordination, and autonomy among spacecraft, satellites, robotic explorers, and, ultimately, human outposts.
The extreme physical conditions of space fundamentally challenge conventional assumptions underlying terrestrial distributed systems.
In particular, the strange predictions of Einstein’s theory of relativity become readily observable due to the vast distances involved, the high relative velocities of communicating entities,
	and heterogeneous gravitational fields.
These factors give rise to relativistic effects such as time dilation and length contraction, which directly impact how time is perceived and measured
	from the reference frame of different system components.

Because relativity predicts that time and simultaneity are relative,
	reasoning about collections of interdependent events becomes substantially more complex.
Consequently, the classical Newtonian model of distributed system executions---where events are totally ordered in real time---no longer
	yields a canonical representation that can be agreed upon by all possible observers.
This is not a far-fetched, science-fiction concern. Hafele and Keating's prediction \cite{Hafele:Predicted} and famous experiment \cite{Hafele:Observed} vividly illustrated this phenomenon as early as the 1970s, without ever leaving the planet: two atomic clocks were flown around the world in opposite directions and compared to a reference clock on the ground, yielding measurable differences in elapsed time due to relativistic effects. If commercial aircraft suffice to produce such discrepancies, it is hardly a stretch to imagine the components of a distributed system finding themselves in comparable or more extreme situations, whether on satellites in distinct orbits, on different planets, or aboard spacecraft moving at high relative velocities. In such settings, the clocks of these components experience time dilation and accumulate discrepancies in local time, challenging the assumption of a shared notion of time that underlies the design and analysis of distributed algorithms.
This, in turn, necessitates careful rethinking of long-cherished concepts such as Herlihy and Wing's linearizability property \cite{her:lin}, which
	states that operations on a shared object appear to take effect in some linear order consistent with their precedence in real time.

To address this fundamental problem, Gilbert and Golab \cite{gilgol} formulated \emph{relativistic linearizability}---a generalization of classic
	(i.e., Herlihy and Wing's) linearizability to executions where events are only partially ordered in a physical sense.
Furthermore, they introduced a collection of theoretical tools to facilitate rigorous analysis of shared object implementations.
While they conjectured that several widely-studied classically linearizable algorithms are also relativistically linearizable,
	and moreover that this point could be proved using their specific technique, their foundational
	work stopped short of presenting formal proofs of correctness.
In response, Jayanti \cite{jayanti} recently proved that all classically linearizable asynchronous algorithms are
	in fact relativistically linearizable.
Intuitively, asynchronous algorithms are immune to the effects of time dilation as they rely only on physical causality
	(by way of program order and message exchange) for coordination,
	which is invariant across different frames of reference.

While it may appear on first impression that Jayanti's theorem confirms Gilbert and Golab's conjecture,
	the two works in fact refer to distinct notions of relativistic linearizability.
Gilbert and Golab formulate three such notions: R1-linearizability as the property that an execution appears linearizable in \emph{some} frame of reference,
	R2-linearizability as a special case of R1 where the execution appears linearizable in \emph{every} frame of reference,
	and R3-linearizability as a special case of R2 where a \emph{single} linearization applies in all frames of reference.
Their conjecture is stated specifically with respect to R3 (which implies R2 and R1), whereas Jayanti's definition
	of relativistic linearizability is equivalent under general relativity to the (strictly) weaker R2 property.
As a result, Jayanti's theorem neither contradicts nor proves the Gilbert and Golab conjecture.

In this paper, we advance the understanding of relativistic distributed systems along several directions.
\begin{enumerate}
	\item We partly confirm Gilbert and Golab's conjecture.
	First, we present a detailed analysis of the Raft replicated state machine \cite{raft}, which can be used to implement a shared object of arbitrary sequential specification, in \Cref{sec:raft}.
	We then analyze the Attiya, Bar-Noy and Dolev (ABD) shared register construction \cite{abd} in \Cref{sec:abd}.
    An informal discussion of quorum-replicated key-value stores follows in \Cref{sec:kv}.
	We establish R3-linearizability in general for Raft, and in special cases for the other algorithms.
	\item Our analysis yields the first proofs of R3-linearizability for concrete distributed algorithms,
		and also the first examples of applying Gilbert and Golab's proof technique \cite{gilgol}.
	\item Whereas Jayanti's theorem \cite{jayanti} allows the linearization order for an algorithm to be relative (i.e., inherently dependent on the observer's frame of reference),
		our results for common asynchronous algorithms demonstrate that all observers can in fact agree on a single global linearization order---a result
		of both practical and theoretical importance.
\end{enumerate}


\newcommand{\compl}{\textrm{compl}}
\newcommand{\optype}{operation type}
\newcommand{\Optype}{Operation type}
\newcommand{\optypetype}{type}
\newcommand{\opex}{operation} 
\newcommand{\Opex}{Operation}
\newcommand{\astep}{atomic step} 
\newcommand{\opop}{operation} 
\newcommand{\Opop}{Operation}
\newcommand{\INV}{\textsc{In}}
\newcommand{\RES}{\textsc{Re}}

\newcommand{\typetupl}{\ensuremath{(\States, \sinit, \Ops, \Resps, \delta)}}

\newcommand{\op}{op}
\newcommand{\ot}{ot}
\newcommand{\ox}{ox}
\newcommand{\re}{ret}
\newcommand{\RefObj}{\O_{\Typ}}
\newcommand{\Done}{OK}
\newcommand{\CRASH}{CRASH}

\newcommand{\E}{\ensuremath{\mathcal{E}}}

\section{Model} \label{sec:model}
Our model is based very closely on \cite{gilgol}.
A set of $N$ processors, denoted $\Procs = \bc{p_1, p_2, ..., p_N}$,
	is subject to relativistic effects due to motion and gravitational fields \cite{einstein}.
Since it is impossible for all observers (i.e., processors in their own frames of references)
	to agree on a total order of events,
	an execution is modeled as a history $H = (E, <_E)$ where $E$ is a finite set of \emph{events}
	and $<_E$ is an irreflexive partial order over $E$.
Each event has a position in Einstein's four-dimensional spacetime, and represents a primitive action (sending/receiving a message, or local computation).
The partial order $<_E$ captures physical causality, which is invariant across frames of reference.
Physical causality is a refinement of Lamport's ``happens before'' relation \cite{lam:ipc1,lam:time}, which captures
	only program order and message passing---the sole means of coordination in purely asynchronous algorithms.
A \emph{classic history} $H = (E, <_E)$ is one where $<_E$ is a total order, meaning that events
	are observed from the frame of reference of a specific processor.
Different processors may perceive the same execution through distinct classic histories with conflicting total orders over the events,
	but all such total orders are refinements of the underlying partial order of physical causality.

The processors execute algorithms that implement shared objects.
Following Herlihy and Wing \cite{her:lin}, processors interact with the shared object implementations
	by invoking operations and receiving their responses.
An invocation event of operation $\op$ by processor $p$ on object $X$ is denoted (\INV$, p, X, \op$).
A response event with return value $\re$ is denoted (\RES$, p, X, \re$).
A response is \emph{matching} with respect to an invocation if both refer to the same processor and object.
An \emph{operation execution} (or op-ex) comprises an invocation and its matching response, if it exists.
For any history $H = (E, <_E)$, an op-ex $\ox$ is \emph{pending} if it lacks a matching response, and is \emph{complete} otherwise.
We denote by $inv(\ox)$ and $res(\ox)$ the invocation and response events of $\ox$, respectively.
If $\ox$ is complete, then $inv(\ox) <_E res(\ox)$.
We denote by $\compl(H)$ the subsequence of $H$ comprising the events of complete op-ex's in $H$.

Given a history $H = (E, <_E)$ and processor $p$, we denote by $H|p$ the subhistory comprising events applied at processor $p$
	and the corresponding subset of $<_E$.
Two histories $H$ and $H'$ are \emph{equivalent} if for every processor $p$, $H|p = H'|p$.
A history $H$ is \emph{well-formed} if for every processor $p$, if $H|p$ is non-empty, then
the events in $H|p$ are totally ordered and form an alternating sequence of invocation and response events,
starting with an invocation.

Lamport \cite{lam:ipc1} defines two temporal relations over pairs of operation executions, $\ox_1$ and $\ox_2$ in a history $H = (E, <_E)$:
$\ox_1 \rightarrow_H \ox_2$ denotes that $\ox_1$ is complete and its response causally precedes the invocation of $\ox_2$; and
$\ox_1 \dashrightarrow_H \ox_2$ denotes that the invocation of $\ox_1$ causally precedes some event of $\ox_2$.
We say that $H$ is \emph{sequential} if every pair of distinct op-ex's is related by $\rightarrow_H$.

The correct behavior of an object under sequential access is defined by its \emph{type} $\Typ = \typetupl$
where $\States$ is the set of states, $\sinit \in \States$ is the
initial state, $\Ops$ is a set of \optype{s},
$\Resps$ is the set of responses, and $\delta: \States \times
\Ops \goesto \States \times \Resps$ is a (one-to-many) state transition mapping.
Specifically, if a processor applies an \opop\ of \optypetype\ $\ot$ to an object
of type $\Typ$ that is in state $s$, then the object may return
a response $r$ and change its state to $s'$ if and only if $(s', r) \in \delta(s, \ot)$.
For a sequential history $H$, we say that an object $X$ \emph{conforms} to its type $\Typ = \typetupl$ in $H$ if
$H|X$ is consistent with some sequence of transitions of $\delta$ starting from state $\sinit$, in the following sense:
letting $\ot_i$ and $\re_i$ denote the operation type and response of the
$i$'th op-ex on $X$ in $H|X$ out of $k$,
there exists a sequence $\bg{s_0, s_1, s_2, ..., s_k}$ of states in $\States$
such that $s_0 = \sinit$, and $(s_i, \re_i) \in \delta(s_{i-1}, \ot_i)$ holds for all $i \leq k$.
We say that $H$ is \emph{legal} if for every object $X$ accessed in $H$, $X$ conforms to its type in $H$.


\section{Background} \label{sec:back}
This section reviews the formal definitions of classic and relativistic linearizability.
We start with classic linearizability as conceived by Herlihy and Wing \cite{her:lin} and explained in \cite{gilgol}:
\begin{definition}\label{def:lin}
	A classic history $H = (E, <_E)$ is \emph{linearizable} if there exists
	a classic history $H' = (E', <_{E'})$ (called the \emph{completion} of $H$) such that:
	\begin{enumerate}
		\item [\textbf{L1}] $H \subseteq H'$ (i.e., $E \subseteq E'$ and $<_E \subseteq <_{E'}$);
		\item [\textbf{L2}] $E'$ is obtained from $E$ by adding a set $M$ of matching responses for a subset of pending operation executions in $H$,
			and each event of $E$ precedes every event of $M$ in $<_{E'}$;
		\item[\textbf{L3}] $\compl(H')$ is equivalent to some legal sequential history $S$; and
		\item[\textbf{L4}] $\rightarrow_S$ extends $\rightarrow_H$, meaning that
			for any pair of operation executions $\ox,\ox'$ in $S$, if $\ox \rightarrow_S \ox'$
			then $\ox'' \rightarrow_H \ox$ is false where $\ox''$ denotes the (possibly pending) counterpart of $\ox'$ in $H$.
	\end{enumerate}
	The sequential history $S$ referred to by property L3 is called a \emph{linearization} (in this paper a \emph{classic linearization}) of~$H$.
\end{definition}

Next, we restate Gilbert and Golab's \cite{gilgol} three variations of relativistic linearizability.
Given an execution history $H = (E, <_E)$ with partially ordered events, the definitions
	realize observations of $H$ from a specific frame of reference using a total order $<_T$ over the events in $E$
	that refines $<_E$ (i.e., $<_E \subseteq <_T$).

\begin{definition}[R1-linearizability]\label{def:r1}
	For any execution history $H = (E, <_E)$, we say that $H$ is \emph{R1-linearizable} if and only if
	there exists a total order $<_T$ over $E$ that refines $<_E$, such that
	the classic history $H_T = (E, <_T)$ is linearizable.
	Any linearization of $H_T$ is called an \emph{R1-linearization of $H$}.
\end{definition}

\begin{definition}[R2-linearizability]\label{def:r2}
	For any execution history $H = (E, <_E)$, we say that $H$ is \emph{R2-linearizable} if and only if
	for every total order $<_T$ over $E$ that refines $<_E$,
	the classic history $H_T = (E, <_T)$ is linearizable.
	Any linearization $S$ of any such $H_T$ is called an \emph{R2-linearization of $H$}.
\end{definition}

\begin{definition}[R3-linearizability]\label{def:r3}
	For any execution history $H = (E, <_E)$, we say that $H$ is \emph{R3-linearizable} if and only if
	there exists an R1-linearization $S$ of $H$ such that
	for every total order $<_T$ over $E$ that refines $<_E$,
	$S$ is a linearization of the classic history $H_T = (E, <_T)$.
	Any such history $S$ is called an \emph{R3-linearization} of $H$.
\end{definition}

Intuitively, R1 states that a history $H$ appears linearizable in \emph{some} frame of reference,
	R2 states that $H$ appears linearizable in \emph{every} frame of reference,
	and R3 states that $H$ is not only linearizable in every frame of reference but
	all observers furthermore agree on a \emph{common linearization}.
R3 implies R2, which implies R1.
R2 is the equivalent of Herlihy and Wing's classic linearizability property in the relativistic model,
	and naturally inherits the \emph{locality property}:
	a history $H$ involving multiple implemented objects is linearizable if and only if the projection
	of $H$ onto each implemented object $O$ (denoted $H|O$) is individually linearizable.
As shown in \cite{gilgol}, R3 is not local.
The proof of this result exhibits a multi-object execution history that is R2-linearizable but not R3-linearizable,
	meaning that every observer perceives classical linearizability and yet
	there exist observers who fundamentally disagree on the linearization order.
Another example of this phenomenon is presented in \cite{jayanti}.
Ideally, a history $H$ would satisfy R2 linearizability and each projection $H|O$ onto
	an object $O$ accessed in $H$ would satisfy the stronger R3 property.

Gilbert and Golab's proof techniques for R2 and R3 linearizability are based on
	the insight that coordination mechanisms used by shared object implementations
	induce certain structural properties on execution histories.
These properties are formalized in terms of \emph{connectedness}, which characterizes causal relationships:

\begin{definition}\label{def:con}
	For any history $H = (E, <_E)$, and any distinct operation executions $\ox,\ox'$ in $H$,
	we say that $\ox$ and $\ox'$ are:
	\begin{itemize}
		\item \emph{strongly connected in $H$} if
		$\ox \dashrightarrow_H \ox'$ and $\ox' \dashrightarrow_H \ox$
		\item \emph{weakly connected in $H$} if either
		$\ox \dashrightarrow_H \ox'$ or $\ox' \dashrightarrow_H \ox$ (but not both)
		\item \emph{connected in $H$} if
		they are strongly or weakly connected in $H$
		\item \emph{disconnected in $H$} if
		they are not connected in $H$
	\end{itemize}
The history $H$ is called \emph{connected} if every pair of operation executions in $H$ is connected.
\end{definition}

Using this notion, Gilbert and Golab show that some R1-linearizable implementations automatically satisfy R2:

\begin{definition}\label{def:r2cond}
	For any history $H = (E, <_E)$, an R1-linearization $S$ of $H$ is called \emph{R2-conducive}
	if for every pair of operation executions $\ox,\ox'$ in $H$ that have counterparts in $S$, the following hold:
	\begin{enumerate}
		\item\label{dr2c1} if $\ox$ and $\ox'$ are weakly connected and $\ox \dashrightarrow_H \ox'$, then $\ox \rightarrow_S \ox'$; and
		\item\label{dr2c2} if $\ox$ and $\ox'$ are disconnected then both are executions of read-only operation types (i.e., ones that always cause trivial state transitions),
		and the same holds for all operation executions that appear between $\ox$ and $\ox'$ in $S$.
	\end{enumerate}
\end{definition}

\begin{theorem}\label{thm:r1r2r3}
	Let $H = (E, <_E)$ be any history and suppose that $S$ is an R2-conducive R1-linearization of $H$.
	Then $H$ is R2-linearizable.
	Furthermore, if $H$ is connected then $S$ itself is an R2-linearization of $H$ (in any frame of reference), hence $H$ is R3-linearizable.\footnote{This result succinctly combines Theorems~5 and 6 in \cite{gilgol}.}
\end{theorem}

Gilbert and Golab conclude their paper by conjecturing that many classic distributed shared object implementations, particularly those based on majority quorums
	(e.g., ABD \cite{abd}, replicated state machines and key-value stores), generate connected execution histories with R2-conducive R1-linearizations,
	and hence satisfy both R2 and R3 linearizability via \Cref{thm:r1r2r3}.


\section{Raft Replicated State Machine} \label{sec:raft}

Raft is a consensus algorithm designed to manage a replicated log across a cluster of servers (i.e., processors in our model), ensuring that each server agrees on the same sequence of state machine commands even in the presence of failures \cite{raft}. Unlike Paxos, which is notoriously difficult to reason about, Raft was explicitly designed for understandability by decomposing the consensus problem into three relatively independent subproblems: leader election, log replication, and safety. At any given time, each server in a Raft cluster is in one of three states---leader, follower, or candidate---and time is divided into terms of arbitrary length, each identified by a monotonically increasing integer. A term begins with an election, and if a leader is successfully elected, it serves for the remainder of that term, coordinating all client interactions and log replication.

The algorithm proceeds roughly as follows. In normal operation, the leader receives client requests, appends them as entries to its local log, and replicates those entries to the other servers in parallel via remote procedure calls (RPCs). Once a majority of servers have acknowledged a given entry, the leader considers it committed and applies it to its state machine. If the leader fails or becomes unreachable, followers that stop receiving heartbeats will time out, transition to the candidate state, and initiate a new election by incrementing the term and requesting votes from their peers. A candidate must receive votes from a majority of the cluster to become the new leader, and Raft's voting rules guarantee that any elected leader already contains all previously committed entries. This majority-based mechanism---used in both elections and commit decisions---is what ensures that committed entries are never lost, even when servers crash and recover.

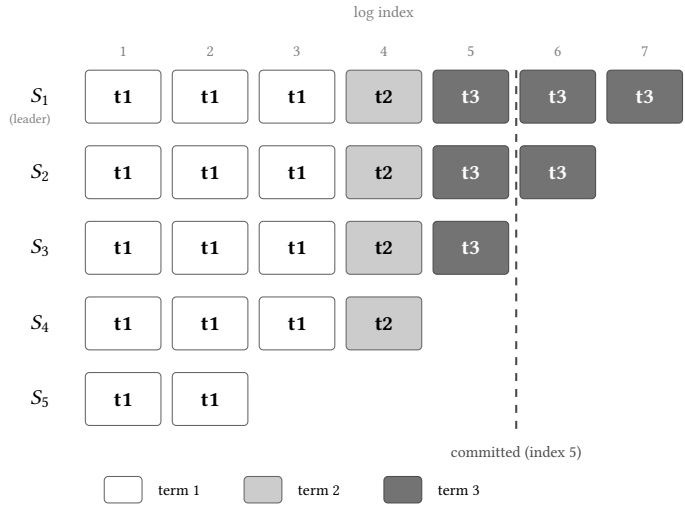
\begin{figure}[t]
\centering
\begin{tikzpicture}[
    entry/.style={minimum width=1.0cm, minimum height=0.7cm, draw, rounded corners=1.5pt, font=\small\bfseries, inner sep=0pt},
    term1/.style={entry, fill=white, draw=black!60},
    term2/.style={entry, fill=black!20, draw=black!60},
    term3/.style={entry, fill=black!55, draw=black!70, text=white},
    server/.style={font=\small\bfseries, anchor=east},
    idx/.style={font=\scriptsize, text=gray},
]

\def\colsep{1.15}
\def\rowsep{1.0}

\foreach \i in {1,...,7} {
    \node[idx] at ({\i*\colsep}, 0.6) {\i};
}
\node[font=\scriptsize, text=gray] at ({4*\colsep}, 1.1) {log index};

\def\row{0}
\node[server] at (0.3, -\row*\rowsep) {$S_1$};
\node[font=\tiny, anchor=east, text=gray] at (0.3, -\row*\rowsep - 0.3) {(leader)};
\foreach \i in {1,2,3} { \node[term1] at ({\i*\colsep}, -\row*\rowsep) {t1}; }
\node[term2] at ({4*\colsep}, -\row*\rowsep) {t2};
\foreach \i in {5,6,7} { \node[term3] at ({\i*\colsep}, -\row*\rowsep) {t3}; }

\def\row{1}
\node[server] at (0.3, -\row*\rowsep) {$S_2$};
\foreach \i in {1,2,3} { \node[term1] at ({\i*\colsep}, -\row*\rowsep) {t1}; }
\node[term2] at ({4*\colsep}, -\row*\rowsep) {t2};
\foreach \i in {5,6} { \node[term3] at ({\i*\colsep}, -\row*\rowsep) {t3}; }

\def\row{2}
\node[server] at (0.3, -\row*\rowsep) {$S_3$};
\foreach \i in {1,2,3} { \node[term1] at ({\i*\colsep}, -\row*\rowsep) {t1}; }
\node[term2] at ({4*\colsep}, -\row*\rowsep) {t2};
\node[term3] at ({5*\colsep}, -\row*\rowsep) {t3};

\def\row{3}
\node[server] at (0.3, -\row*\rowsep) {$S_4$};
\foreach \i in {1,2,3} { \node[term1] at ({\i*\colsep}, -\row*\rowsep) {t1}; }
\node[term2] at ({4*\colsep}, -\row*\rowsep) {t2};

\def\row{4}
\node[server] at (0.3, -\row*\rowsep) {$S_5$};
\foreach \i in {1,2} { \node[term1] at ({\i*\colsep}, -\row*\rowsep) {t1}; }

\draw[dashed, thick, black!70]
    ({5*\colsep + 0.6}, 0.35) -- ({5*\colsep + 0.6}, -4.4);
\node[font=\scriptsize, text=black!70, anchor=north] at ({5*\colsep + 0.6}, -4.5) {committed (index 5)};

\node[term1, minimum width=0.5cm, minimum height=0.35cm] at (1.15, -5.2) {};
\node[font=\scriptsize, anchor=west] at (1.5, -5.2) {term 1};
\node[term2, minimum width=0.5cm, minimum height=0.35cm] at (3.0, -5.2) {};
\node[font=\scriptsize, anchor=west] at (3.35, -5.2) {term 2};
\node[term3, minimum width=0.5cm, minimum height=0.35cm] at (4.85, -5.2) {};
\node[font=\scriptsize, anchor=west] at (5.2, -5.2) {term 3};

\end{tikzpicture}
\caption{Replicated log structure across a five-server Raft cluster during term~3. Each cell represents a log entry labeled with the term in which it was created. The dashed line marks the commit index: entries through index~5 have been replicated to a majority and are considered committed.}
\label{fig:raft-log}
\end{figure}

\Cref{fig:raft-log} depicts the structure of Raft's log. Each cell represents a log entry (i.e., a command) labeled with the term in which it was created. Row $i$ represents the local log of server $S_i$. The leader's local log ($S_1$ in \Cref{fig:raft-log}) has the largest number of entries in its leadership term, as it is the first node (i.e., server) that receives the clients' commands. Whenever an entry is logged into the leader's local log, the leader replicates it to other servers concurrently, and once the majority of the servers (in this example, at least 2 follower servers) acknowledge the entry, it will be considered as committed. In \Cref{fig:raft-log}, 5 commands have been committed, which means in the case of failure, any server that would become the next leader will have the first five entries in the log.

The Raft algorithm carries an inherent causal relationship among its operations.
This can be observed from the algorithm's behavior, where the state of the cluster is determined by the operations of nodes. E.g., a leader becomes unavailable, which causes some other node to identify itself as a candidate, which results in a new term, and possibly a new leader. Moreover, a majority of servers must acknowledge an entry before it is committed and applied to the state machine.
This kind of behavior of the algorithm leads us to believe that a unique linearization order exists over the operations, as determined by the unique positions of the corresponding commands in the replicated log. Therefore, we expect Raft to be R3-linearizable. However, this observation does not guarantee that R3-linearizability of Raft can be proved by the technique introduced by Gilbert and Golab, and Jayanti's theorem proves Raft R2-linearizability but does not settle the question for R3-linearizability.

\begin{figure}[t]
\centering
\begin{tikzpicture}[scale=1]
 
\def\ya{0}      
\def\yb{-1.5}   
\def\ycc{-3.0}  
 
\draw[-{Stealth[length=3pt]}, gray!40, thin]
    (0, 0.7) -- (7.8, 0.7) node[right, font=\scriptsize, text=gray] {time};
 
\node[plabel] at (-0.2, \ya) {$S_1$};
\node[font=\tiny, anchor=north east, text=gray] at (-0.15, \ya-0.12) {(leader)};
\node[plabel] at (-0.2, \yb) {$S_2$};
\node[plabel] at (-0.2, \ycc) {$S_3$};
 
\draw[timeline] (0, \ya) -- (2.2, \ya);
\draw[timeline dashed] (2.2, \ya) -- (7.5, \ya);
 
\draw[timeline] (0, \yb) -- (7.5, \yb);
\draw[timeline] (0, \ycc) -- (7.5, \ycc);
 
 
\draw[msg, black!50] (0.4, \ya) -- (1.1, \yb);
\node[ev, black!50] at (0.4, \ya) {};
\node[ev, black!50] at (1.1, \yb) {};
 
\draw[msg, black!50] (0.4, \ya) -- (1.3, \ycc);
\node[ev, black!50] at (1.3, \ycc) {};
 
\node[mlabel, text=black!60, anchor=east] at (0.55, {(\ya+\yb)/2}) {heartbeat};
 
 
\node[font=\bfseries, text=black] at (2.2, \ya) {$\times$};
\node[mlabel, text=black!70, anchor=south] at (2.2, \ya+0.17) {crash};
 
 
\draw[gray!40, thin, dashed] (3.5, \yb+0.25) -- (3.5, \yb-0.25);
\node[mlabel, text=gray, anchor=south, align=center] at (3.5, \yb+0.3)
    {election\\[-2pt]timeout};
 
\node[ev, black!80] at (3.7, \yb) {};
 
\node[mlabel, text=black!70, anchor=north] at (3.7, \yb-0.2)
    {\itshape term++};
 
 
\draw[msg lost, black!40] (3.9, \yb) -- (4.6, \ya);
\node[ev, black!40, opacity=0.4] at (4.6, \ya) {};
 
\node[font=\scriptsize, text=black!60, opacity=0.6] at (4.75, \ya) {$\times$};
 
\draw[msg, black!80] (3.9, \yb) -- (4.7, \ycc);
\node[ev, black!80] at (3.9, \yb) {};
\node[ev, black!80] at (4.7, \ycc) {};
 
\node[mlabel, text=black!80, anchor=east] at (4.1, {(\yb+\ycc)/2})
    {\textsc{RequestVote}};
 
 
\draw[msg, black!80, densely dashed] (4.9, \ycc) -- (5.7, \yb);
\node[ev, black!80] at (4.9, \ycc) {};
\node[ev, black!80] at (5.7, \yb) {};
 
\node[mlabel, text=black!80, anchor=west] at (5.5, {(\yb+\ycc)/2})
    {\textsc{VoteGranted}};
 
 
\draw[black!70, thick, rounded corners=1pt]
    (5.62, \yb-0.12) rectangle (5.78, \yb+0.12);
\node[mlabel, text=black!70, anchor=south, align=center]
    at (5.7, \yb+0.17) {$S_2$ elected\\[-2pt]leader};
 
\node[mlabel, text=black!60, anchor=north west, align=left]
    at (5.85, \yb-0.25) {\itshape majority ($S_2 + S_3$)};
 
\draw[msg, black!80] (0.0, -4.0) -- (0.5, -4.0);
\node[mlabel, anchor=west] at (0.55, -4.0) {RPC request};
\draw[msg, black!80, densely dashed] (2.3, -4.0) -- (2.8, -4.0);
\node[mlabel, anchor=west] at (2.85, -4.0) {RPC response};
\draw[msg lost, black!40] (4.8, -4.0) -- (5.3, -4.0);
\node[font=\scriptsize, text=black!60, opacity=0.6] at (5.45, -4.0) {$\times$};
\node[mlabel, anchor=west] at (5.6, -4.0) {lost/no reply};
 
\end{tikzpicture}
\caption{Leader failure and election in Raft. After $S_1$ crashes, $S_2$'s
election timer expires and it transitions to the candidate state, incrementing
its term. $S_2$ sends \textsc{RequestVote} RPCs to the other servers; $S_1$
does not respond (dashed arrow), but $S_3$ grants its vote. With a majority
of votes ($S_2 + S_3$), $S_2$ becomes the new leader for the current term.}
\label{fig:raft-election}
\end{figure}
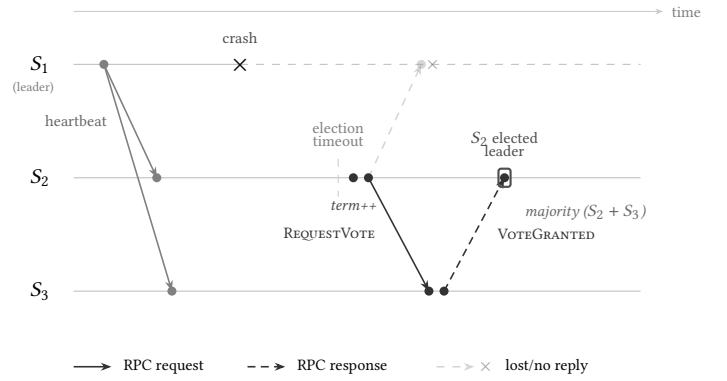

\Cref{fig:raft-election} represents an execution in which the leader (server $S_1$) crashes and fails to send heartbeat RPCs. When the election timeout is reached, $S_2$ increments the term and requests votes from other servers, in order to become the new leader.
The vote is granted to $S_2$ by $S_3$ and $S_2$ becomes the new leader.\footnote{$S_2$ votes for itself, hence only one additional vote is required to form a majority in the three-server case.}
Raft's voting scheme ensures that the elected server has a log that is at least as up-to-date as that of every server in some majority quorum.
This is guaranteed because when a server (in this example $S_2$) requests votes from other servers by sending RPCs to them, it also attaches the index and term of its last log entry along with the new term number to the RPC, and a voter endorses the candidate only if the candidate's log is at least as up-to-date as its own.
Once $S_2$'s election is confirmed, the clients send their commands to $S_2$, and it coordinates their replication.

One can easily see the many causal relationships among the events in \Cref{fig:raft-election}.
For example, $S_2$ would not have increased the term and started a leader election if $S_1$ had not stopped sending heartbeats, at least assuming a timeout value appropriately tuned to the network and processing delays. Note that the heartbeat timeout is measured locally by each server, and therefore is immune to the time dilation and length contraction effects.

The leader transition exemplifies even more causal relationships in Raft, independently of the chosen heartbeat timeout.
In the scenario of \Cref{fig:raft-election}, for $S_2$ to become the leader after $S_1$, the voting scheme requires $S_2$'s log to already contain any command corresponding to $S_1$'s term and earlier terms if that command was already committed or will be committed in the future.
If not, then the majority of the servers that voted for $S_2$ lack the command and will no longer accept it after casting their votes
for $S_2$ because the Raft protocol rejects RPCs from obsolete terms.
Raft's \emph{log matching} property then ensures that any such command is replicated to a majority of servers before $S_2$ commits any new command in its own term.
This enforces causality between the operations corresponding to commands committed in earlier terms and commands committed in $S_2$'s term.
Thus, the causal relationships yield a unique linearization order, recorded in some prefix of the leader's local log, and independent of the frame of reference.
Moreover, this committed prefix (equivalently, the state machine state it produces) serves as a canonical record of the execution that every observer agrees upon, regardless of the frame of reference from which it is observed.
Hence, we expect the Raft algorithm to be R3-linearizable.

In order to prove R3-linearizability formally using the Gilbert-Golab technique, we first show that all pairs of operations are connected, then exhibit an R2-conducive R1-linearization, and finally apply \Cref{thm:r1r2r3}.
Consider any history $H = (E, <_E)$ of Raft.
The completion $H' = (E', <_{E'})$ of $H$ is constructed by completing the pending op-exes with committed commands, and removing all other pending ones from $H$.
For completed op-exes, the matching response event carries a return value determined by simulating the execution of the corresponding command based on the sequence of commands in the replicated log.
A unique coordinate in spacetime is assigned to each matching response event $e$ such that every event of $H$ causally precedes $e$.
Next, we construct a \emph{canonical R1-linearization} of $H$ from $H'$.
Fix a total order $<_T$ that refines $<_{E'}$ and corresponds to the reference frame of the final Raft leader (i.e., the leader for the maximum term that occurs in $H$).\footnote{This server is guaranteed to have the most complete copy of the replicated log.}
Let $H_T = (E', <_T)$, and let $S$ denote a sequential ordering of operations in $H'$ based on the ultimate order of the corresponding commands in the final leader's local log.
In other words, if the op-exes $\ox$ and $\ox'$ both appear in $H'$ and the command for $\ox$ precedes the command for $\ox'$
	in the leader's log at the end of $H$, then $\ox \rightarrow_S \ox'$.
Since commands for op-exes in $H'$ are all committed at a majority of servers in Raft in $H$, it follows that
	$\rightarrow_S$ is consistent with $\rightarrow_{H_T}$.
Furthermore, the history $S$ is legal since it reflects the final leader's execution of commands in the correct order.
Thus, $S$ is an R1-linearization of $H$.

\begin{lemma} \label{lem:raft:connect}
Let $H = (E, <_E)$ be any history of Raft and let $H'$ be its completion.
Then, any two operation executions $\ox$ and $\ox'$ in $H$ that have counterparts in $H'$ are connected.
\end{lemma}
\begin{proof}
 Let $\ox$ and $\ox'$ be operation executions in $H$ representing state machine commands $cmd$ and $cmd'$, respectively.
 Since we assume both op-exes have counterparts in $H'$, it follows that both have corresponding commands committed to a majority of servers.
 The two commands are committed either in the same leadership term or during distinct ones.

 Suppose that the two commands are committed in the same term, and without the loss of generality, $cmd$ precedes $cmd'$ in the replicated log.
 Therefore, the event in which the leader appends $cmd$ causally precedes the one appending $cmd'$ due to the leader's program order.
 This implies that $\ox \dashrightarrow_{H'} \ox'$, as required.

 The other case to consider is when the two commands are executed during distinct terms.
 Without the loss of generality assume $cmd$ and $cmd'$ are in terms $t$ and $t'$, respectively, and $t < t'$.
 The leader election at the beginning of term $t'$ uses majority voting and establishes a causal relationship between operation executions in different terms, as explained earlier in this section.
 Specifically, the local log of the term $t'$ leader records all committed commands for term $t$ before any committed command for term $t'$.
 Furthermore, Raft's log matching property ensures that all the term $t$ commands are committed before any term $t'$ commands.
 This implies that $\ox \dashrightarrow_{H'} \ox'$, as required.
\end{proof}

\begin{lemma} \label{lem:raft:r2conducive}
Let $H = (E, <_E)$ be any history of Raft, let $H'$ be its completion, and let $S$ be the corresponding canonical R1-linearization of $H$.
Then, $S$ is R2-conducive.
\end{lemma}
\begin{proof}
	Consider two op-exes $\ox$ and $\ox'$ in $H$ that have counterparts in $H'$.
    If the command $cmd$ of an op-ex $\ox$ appears before the command $cmd'$ of an op-ex $\ox'$ in the replicated log, then some event of $\ox$ precedes some event of $\ox'$
    and $\ox \dashrightarrow_H \ox'$ holds, as explained in the proof of \Cref{lem:raft:connect}.
    Thus, the two op-exes are either weakly or strongly connected.
    In both cases, $S$ linearizes the op-exes based on the order of appearance of the corresponding commands in the replicated log.
    If $\ox$ and $\ox'$ are only weakly connected then $\ox \dashrightarrow_H \ox'$ implies that $cmd$ precedes $cmd'$ in the log, hence $\ox \rightarrow_S \ox'$.
    Therefore, the first clause of \Cref{def:r2cond} holds for $S$.
    From \Cref{lem:raft:connect} we know that any two op-exes in $H'$ are connected, and therefore, the second clause of \Cref{def:r2cond} holds trivially.
    Thus, the lemma is proved.
\end{proof}

\begin{theorem} \label{lem:raft:r3lin}
	Let $H = (E, <_E)$ be any history of Raft, let $H'$ be its completion, and let $S$ be the corresponding canonical R1-linearization of $H$.
	Then, $H$ is R3-linearizable, and $S$ is its R3-linearization.
\end{theorem}
\begin{proof}
        It follows from \Cref{lem:raft:r2conducive} that $S$ is R2-conducive. Hence, by \Cref{thm:r1r2r3}, $H$ is R2-linearizable.
        Moreover, from \Cref{lem:raft:connect} we know that $H'$ is connected.
        Hence, by \Cref{thm:r1r2r3}, $S$ is an R2-linearization of $H$ in any frame of reference.
        Thus, $H$ is R3-linearizable and $S$ is its R3-linearization.
\end{proof}

\section{ABD Simulation} \label{sec:abd}

The ABD simulation (\Cref{fig:abd}) implements a classically linearizable single-writer multi-reader shared register object.
The state of the register is replicated across the $N$ processors using majority quorums.
The values written to the register are tagged with monotonically increasing version numbers, generated by the writer.
A write operation sends the new value-version pair to all processors and awaits a reply from a majority.
In case messages are reordered by the network, each processor compares the incoming value-version pair against
	its local state, and overwrites its previous value only if the incoming one has a larger version.
A read operation sends a request to retrieve value-version pairs from all processors and awaits replies from a majority.
If all the reply messages have equal versions, then their values agree and the common value is returned.
If not, then the reader executes a write-back phase whereby the received value with the highest version is sent
	to all processors (as in a write).
The read returns this same value after a majority of processors acknowledges receipt of the write-back message.
The construction can tolerate the permanent crash failure of any minority of processors.

\begin{figure}[htbp]
	\removelatexerror
	\SetKwProg{Fn}{Procedure}{}{end}
	\begin{flushleft}
	\textbf{Variables:}
	\begin{itemize}
	\item $last\_version$: integer, initially zero (writer only)
	\item $val_i$, $version_i$: processor $i$'s value (initially equal to the initial value of the simulated read/write register)
		  and its corresponding version (initially zero)
	\end{itemize}
	\end{flushleft}
	\begin{algorithm}[H]
		 \Fn{Write(v) for writer processor}{
		 	$last\_version := last\_version + 1$\;
		 	send $\bg{v, last\_version}$ to all processors\;
		 	wait for ACK from majority
		}
	\end{algorithm}
	\SetKwProg{Fn}{Procedure}{}{end}
	\begin{algorithm}[H]
		\Fn{Read() for processor $i$}{
			request $\bg{val_j, version_j}$ from each processor $j$\;
			$R := $ responses from majority of processors\;
			$\bg{v, version} := $ element of $R$ with largest version\;
			\If {$|R| > 1$} {
				send $\bg{v, version}$ to all processors\;
				wait for ACK from majority\;
			}
			\Return $v$\;
		}
	\end{algorithm}
	\SetKwProg{Fn}{Procedure}{}{end}
	\begin{algorithm}[H]
		\Fn{Receive $\bg{v, version}$, processor $i$}{
			\If {$version > version_i$} {
				$val_i := v$\;
				$version_i := version$\;
			}
			send back ACK\;
		}
	\end{algorithm}
	\SetKwProg{Fn}{Procedure}{}{end}
	\begin{algorithm}[H]
	\Fn{Receive request for $\bg{val_i, version_i}$, processor $i$}{
		send back $\bg{val_i, version_i}$\;
	}
	\end{algorithm}

	\caption{Simplified pseudocode for the ABD simulation \cite{abd}.}\label{fig:abd}
\end{figure}

Consider any history $H = (E, <_E)$ of ABD.
The completion $H' = (E', <_{E'})$ of $H$ is constructed as follows.
First, discard any pending reads, as well as any pending write whose value is not returned by any read operation.
If a pending write in $H$ has its value\footnote{
	For analysis, we assume that each read is mapped to a unique write via the version number even if two writes may assign the same value $v$.}
	returned by a read, then ``append'' a matching response.
A unique coordinate in spacetime is assigned to each matching response event $e$ such that every event of $H$ causally precedes $e$.
Next, fix a total order $<_T$ that refines $<_{E'}$.
Let $H_T = (E', <_T)$, and let $S$ denote a sequential ordering of operations in $H_T$ according to the following rules.
First, arrange writes in increasing order of the corresponding versions (i.e., in program order).
Next, let $R_k$ denote the set of read operations that return the value $v$ obtained from a value-version pair $\bg{v, k}$.
Insert each read in $R_k$ after the write operation that creates version $k$ and before the next write.
Within the set $R_k$, the read operations can be ordered in an arbitrary manner consistent with	$\rightarrow_{H_T}$.
Classical analysis of the ABD algorithm shows that the constructed linearization $S$
	is a legal sequential history and a linearization of $H_T$, which implies R1-linearizability of $H$.


To deduce R2 and R3-linearizability from R1, we establish additional structural properties of the history $H$.
Intuitively, the use of majority quorums guarantees that the completion $H'$ is connected
	for any history $H$ of ABD:

\begin{lemma}\label{lem:abd:connected}
	Let $H = (E, <_E)$ be a history of ABD, and let $H'$ be its completion.
	Then any two operation executions $\ox_1$ and $\ox_2$ in $H$ that have counterparts in $H'$ are connected.
\end{lemma}
\begin{proof}
	We proceed by an exhaustive case analysis on the possible pairs of operations in $H'$.

	\noindent \underline{Case~A:} $\ox_1$ and $\ox_2$ are both reads.
	Since only complete reads are included in $H'$, both operations access a majority quorum of processors.
	Since the majority quorums intersect, both operations access some common processor $p$, which establishes causality between
	some event of $\ox_1$ and some event of $\ox_2$.

	\noindent \underline{Case~B:} $\ox_1$ and $\ox_2$ are both writes.
	Since there is only a single writer, both operations are executed by the same processor.
	Once again, this establishes causality.

	\noindent \underline{Case~C:} $\ox_1$ is a read and $\ox_2$ is a write.
	If the write $\ox_2$ is complete in $H$ then both operations are complete in $H$.
	In that case both operations access a majority of processors, which establishes causality as in Case~A.
	On the other hand, if $\ox_2$ is pending, then by construction of $H'$ it follows that
	all events of $\ox_1$ causally precede the response of $\ox_2$.
	Thus, $\ox_1 \dashrightarrow_{H'} \ox_2$.
\end{proof}

To complete the proof of R3-linearizability, it suffices to show that the R1-linearization $S$
is R2-conducive and then apply \Cref{thm:r1r2r3}.
We demonstrate this first for the single-reader case in \Cref{sec:abd:swsr}, and then for the multi-reader case in \Cref{sec:abd:swmr}.

\subsection{Single-Reader Case}\label{sec:abd:swsr}
For simplicity, we first complete the analysis in the special case where only a single processor can issue read operations.\footnote{
Since $N \geq 3$ must hold for fault tolerance, the single reader assumption implies that at least one processor participates in the replication protocol
	but does not invoke operations on the implemented read/write register.}
Intuitively, we expect the ABD algorithm instantiated in this setting to satisfy R3-linearizability
	since there can only be one valid linearization for a history $H$ in any frame of reference (once the completion $H'$ is fixed)
	as the ambiguity regarding the relative order of reads in our construction of $S$ is removed.

\begin{lemma}\label{lem:abd:swsr:conducive}
	Let $H = (E, <_E)$ be a history of single-reader ABD, and let $H'$ be its completion.
	Fix a total order $<_T$ that refines $<_E$, and let $S$ be the corresponding R1-linearization.
	Then for every pair of operation executions $\ox_1$ and $\ox_2$ in $H'$, if $\ox_1$ and $\ox_2$ are weakly connected
	with $\ox_1 \dashrightarrow_{H'} \ox_2$, then $\ox_1 \rightarrow_S \ox_2$.
\end{lemma}
\begin{proof}
Let $\ox_1$ and $\ox_2$ be two op-exes in $H'$ that are weakly connected
	and satisfy $\ox_1 \dashrightarrow_{H'} \ox_2$.
To show that $\ox_1 \rightarrow_S \ox_2$, we proceed by an exhaustive case analysis on the possible pairs of operations in $H'$.

\noindent \underline{Case~A:} $\ox_1$ and $\ox_2$ are both writes.
Since there is only a single writer, both operations are executed by the same processor.
Then $\ox_1 \dashrightarrow_{H'} \ox_2$ implies that $res(\ox_1) <_E inv(\ox_2)$,
	hence $\ox_1 \rightarrow_S \ox_2$ by construction of $S$ irrespective of $<_T$.

\noindent \underline{Case~B:} $\ox_1$ and $\ox_2$ are both reads.
The analysis is analogous to Case~A since there is only one reader.

\noindent \underline{Case~C:} $\ox_1$ is a read and $\ox_2$ is a write.
We proceed by subcases on the relative position of $\ox_1$ and $\ox_2$ in the linearization order.

\noindent \underline{Subcase~(i):} $\ox_1$ returns the value assigned by a write that precedes $\ox_2$ in $S$.
Then $\ox_1 \rightarrow_S \ox_2$ by construction of $S$.

\noindent \underline{Subcase~(ii):} $\ox_1$ returns the exact value $v$ written by $\ox_2$.
Then $\ox_2 \dashrightarrow_{H'} \ox_1$ holds since $\ox_1$ must retrieve $v$ from
	some processor $p$ after $p$ receives it in a message generated by $\ox_2$.
This is a contradiction since the lemma assumes that $\ox_1$ and $\ox_2$ are weakly connected with $\ox_1 \dashrightarrow_{H'} \ox_2$.

\noindent \underline{Subcase~(iii):} $\ox_1$ returns the value $v$ assigned by a write that follows $\ox_2$ in $S$.
Let $\ox_3$ be that write.
Then $\ox_3 \dashrightarrow_{H'} \ox_1$ holds since $\ox_1$ must retrieve $v$ from
some processor $p$ after $p$ receives it in a message generated by $\ox_3$.
Furthermore, $\ox_2$ is complete and $\ox_2 \rightarrow_H \ox_3$.
This implies that $\ox_2 \dashrightarrow_{H'} \ox_1$.
We obtain a contradiction as in Subcase~(ii).

\noindent \underline{Case~D:} $\ox_1$ is a write and $\ox_2$ is a read.
We proceed by subcases on the relative position of $\ox_1$ and $\ox_2$ in the linearization order.

\noindent \underline{Subcase~(i):} $\ox_2$ returns the value $v$ assigned by a write that precedes $\ox_1$ in $S$.
If $\ox_1$ is pending in $H$ then $inv(\ox_2) <_{E'} res(\ox_1)$ by construction of $H'$.
Thus, $\ox_2 \dashrightarrow_{H'} \ox_1$ holds, which is a contradiction since the lemma assumes that $\ox_1$ and $\ox_2$
	are weakly connected with $\ox_1 \dashrightarrow_{H'} \ox_2$.
On the other hand, if $\ox_1$ is complete in $H$, then there is a processor $p$ at the intersection of the quorums
	used by $\ox_1$ and $\ox_2$ such that $\ox_2$ reads $p$'s value before $\ox_1$ overwrites it.
Then $\ox_2 \dashrightarrow_{H} \ox_1$ holds, hence $\ox_2 \dashrightarrow_{H'} \ox_1$,
	which once again leads to a contradiction.

\noindent \underline{Subcase~(ii):} $\ox_2$ returns the value written by $\ox_1$.
Then $\ox_1 \rightarrow_S \ox_2$ by construction of $S$.

\noindent \underline{Subcase~(iii):} $\ox_2$ returns the value assigned by a write that follows $\ox_1$ in $S$.
Let $\ox_3$ be that write.
Then $\ox_1 \rightarrow_S \ox_3$ and $\ox_3 \rightarrow_S \ox_2$ by construction.
This implies $\ox_1 \rightarrow_S \ox_2$.
\end{proof}

\begin{theorem}\label{thm:abd:rlin}
	Let $H = (E, <_E)$ be a history of single-reader ABD, and let $H'$ be its completion.
	Fix any total order $<_T$ that refines $<_E$, and let $S$ be the corresponding R1-linearization.
	Then $H$ is R3-linearizable, and $S$ its R3-linearization.
\end{theorem}
\begin{proof}
The completion $H'$ is connected by \Cref{lem:abd:connected},
	and the R1-linearization $S$ is R2-conducive by \Cref{lem:abd:swsr:conducive}.
R3-linearizability with respect to $S$ follows by \Cref{thm:r1r2r3}.
\end{proof}

\subsection{Multi-Reader Case, $N=3$}\label{sec:abd:swmr}
In the presence of multiple readers, the construction of the linearization $S$ becomes more complex as the
	set $R_k$ of read operations that return the value $v$ obtained from a value-version pair $\bg{v, k}$
	is no longer totally ordered by physical causality via program order.
Since distinct observers may experience Herlihy and Wing's ``happens before'' relation differently due to relativity of simultaneity,
	one observer may perceive two operations as overlapping in time while
	another observer perceives one operation producing a response before the other is invoked.
For R3-linearizability, the operations in the set $R_k$ must therefore be ordered in a manner that respects
	Herlihy and Wing's ``happens before'' relation in \emph{every} conceivable frame of reference.

We address this problem by fixing the order of reads in $R_k$ so that
	if $\ox_1,\ox_2 \in R_k$ are weakly connected and $\ox_1 \dashrightarrow_{H'} \ox_2$ then $\ox_1 \rightarrow_S \ox_2$.
To see why this is sufficient, consider the different ways in which $\ox_1$ and $\ox_2$ can be related by causality.
If $\ox_1$ and $\ox_2$ are strongly connected, then they appear concurrent in all frames of reference, and can be ordered arbitrarily in $S$.
On the other hand, if $\ox_1$ and $\ox_2$ are only weakly connected, then each observer either perceives the two operations
	as concurrent or perceives $\ox_1$ happening before $\ox_2$ since $\ox_1 \dashrightarrow_{H'} \ox_2$.
In both cases, the ordering $\ox_1 \rightarrow_S \ox_2$ is consistent with the observer's subjective perception of time.

It remains to show that the linearization $S$ exists, which is tantamount to the statement that the $\dashrightarrow_{H'}$ relation cannot create a cycle
	when applied to weakly connected op-ex pairs in a given set $R_k$.
We present the proof below for $N=3$.
The result is nontrivial since $\dashrightarrow_{H'}$ is not a transitive relation in the general context, in contrast to Herlihy and Wing's ``happens before'' relation.

\begin{lemma}\label{lem:abd:swmr:acyclic}
	Let $H = (E, <_E)$ be a history of ABD, and let $H'$ be its completion.
	For any $k$, let $G_k$ be the directed graph whose vertices are the read operation executions in $R_k$ and
		where an edge exists from $\ox_1$ to $\ox_2$ if and only if they are weakly connected with $\ox_1 \dashrightarrow_{H'} \ox_2$.
	Then $G_k$ is acyclic for all $k$.
\end{lemma}
\begin{proof}
Suppose for contradiction that for some $k$ the graph $G_k$ has a cycle $C$.
The cycle must have at least three edges since an operation execution cannot be weakly connected to itself (in a one-cycle or loop),
	and similarly two op-exes cannot be weakly connected to each other.
Let $\ox_1$, $\ox_2$, and $\ox_3$ be three consecutive read operation executions along the cycle such that
	$\ox_1 \dashrightarrow_{H'} \ox_2$ and $\ox_2 \dashrightarrow_{H'} \ox_3$,
	and each consecutive pair is weakly connected.

Without loss of generality, assume that $C$ is a cycle of minimal length in $G_k$.
Next, note that the operations in $C$ must be executed by distinct processors,
	as otherwise $C$ can be shortened.
That is, if $\ox_i$ and $\ox_j$ are executed by the same processor and $\ox_i$ precedes $\ox_j$ in program order,
	then letting $\ox_h$ denote the predecessor of $\ox_i$ along $C$,
	$\ox_h$ and $\ox_j$ are weakly connected with $\ox_h \dashrightarrow_{H'} \ox_j$,
	which makes it possible to remove $\ox_i$ (and any op-exes strictly between $\ox_i$ and $\ox_j$) from the cycle.
To see this, first note that because $\ox_h \dashrightarrow_{H'} \ox_i$ and $\ox_i \rightarrow_{H'} \ox_j$,
	$\ox_h \dashrightarrow_{H'} \ox_j$ holds.
Furthermore, $\ox_j \dashrightarrow_{H'} \ox_h$ is false, as otherwise
	$\ox_i \rightarrow_{H'} \ox_j$ implies $\ox_i \dashrightarrow_{H'} \ox_h$,
	which contradicts $\ox_h$ and $\ox_i$ being weakly connected with $\ox_h \dashrightarrow_{H'} \ox_i$.

Without loss of generality, suppose that processes are numbered so that $p_i$ executes $\ox_i$ for $i \in \bc{1, 2, 3}$.
Since we have shown that $C$ has length at least three, and that all operations in the cycle are
	executed by distinct processors, it follows from the assumption $N=3$ that
	$C$ has length exactly three.
Consequently, $\ox_3$ and $\ox_1$ are weakly connected with $\ox_3 \dashrightarrow_{H'} \ox_1$.

Now consider the causal relationships among $\ox_1$, $\ox_2$, and $\ox_3$.
Due to the use of majority quorums in ABD, $\ox_1$ and $\ox_2$ must both fetch value-version pairs
	from a common processor $p$.
Let $e^p_1$ be the event where $p$ receives the value request sent by $\ox_1$,
	and let $e^p_2$ be the event where $p$ receives the value request sent by $\ox_2$.
Since $p$ is sequential, $e^p_1$ and $e^p_2$ are related by $<_E$ via program order.
Moreover, it follows that $e^p_1 <_E e^p_2$
	as otherwise there is a chain of causality $inv(\ox_2) <_E e^p_2 <_E e^p_1 <_E res(\ox_1)$ that contradicts
	$\ox_1$ and $\ox_2$ being weakly connected with $\ox_1 \dashrightarrow_{H'} \ox_2$.
Thus, $inv(\ox_1) <_E e^p_1 <_E e^p_2 <_E res(\ox_2)$ holds.

Next, consider the action of $\ox_3$.
Due to the use of majority quorums in ABD, $\ox_3$ must fetch a value-version pair from at least one other processor $q$.
Since $N=3$, this processor $q$ must be $p_1$ or $p_2$.
Let $e^q_1$ be the event where $q$ receives the value request sent by $\ox_3$,
and let $e^q_2$ be the event where $q$ sends the response received by $\ox_3$.
Observe that $inv(\ox_3) <_E e^q_1 <_E e^q_2 <_E res(\ox_3)$ holds.

\noindent\underline{Case A:} $q = p_1$, hence $e^q_1$ and $e^q_2$ are both related by program order to $\ox_1$.
It follows that $e^q_2 <_E inv(\ox_1)$ holds, as otherwise $inv(\ox_1) <_E e^q_2 <_E res(\ox_3)$,
	which contradicts $\ox_1$ and $\ox_3$ being weakly connected with $\ox_3 \dashrightarrow_{H'} \ox_1$.
Thus, $inv(\ox_3) <_E e^q_1 <_E e^q_2 <_E inv(\ox_1)$ holds.
Since we showed earlier that $inv(\ox_1) <_E e^p_1 <_E e^p_2 <_E res(\ox_2)$ holds,
	it follows from the transitivity of $<_E$ that $inv(\ox_3) <_E res(\ox_2)$.
This contradicts $\ox_2$ and $\ox_3$ being weakly connected with $\ox_2 \dashrightarrow_{H'} \ox_3$.

\noindent\underline{Case B:} $q = p_2$, hence $e^q_1$ and $e^q_2$ are both related by program order to $\ox_2$.
It follows that $res(\ox_2) <_E e^q_2$ holds, as otherwise $inv(\ox_3) <_E e^q_1 <_E e^q_2 <_E res(\ox_2)$,
	which contradicts $\ox_2$ and $\ox_3$ being weakly connected with $\ox_2 \dashrightarrow_{H'} \ox_3$.
Thus, $res(\ox_2) <_E e^q_2 <_E res(\ox_3)$ holds.
Since we showed earlier that $inv(\ox_1) <_E e^p_1 <_E e^p_2 <_E res(\ox_2)$ holds,
	it follows from the transitivity of $<_E$ that $inv(\ox_1) <_E res(\ox_3)$.
This contradicts $\ox_1$ and $\ox_3$ being weakly connected with $\ox_3 \dashrightarrow_{H'} \ox_1$.
\end{proof}

We now complete the analysis of R3-linearizability, noting that \Cref{lem:abd:swmr:acyclic}
	implies the existence of the R1-linearization $S$ for the completion $H'$ of $H$.

\begin{lemma}\label{lem:abd:swmr:conducive}
	Let $H = (E, <_E)$ be a history of ABD for $N=3$ processors, and let $H'$ be its completion.
	Fix a total order $<_T$ that refines $<_E$, and let $S$ be the corresponding R1-linearization.
	Then for every pair of operation executions $\ox_1$ and $\ox_2$ in $H'$, if $\ox_1$ and $\ox_2$ are weakly connected
	with $\ox_1 \dashrightarrow_{H'} \ox_2$, then $\ox_1 \rightarrow_S \ox_2$.
\end{lemma}
\begin{proof}
	We proceed by an exhaustive case analysis on the possible pairs of operations in $H'$.
	The case analysis is identical to the proof of \Cref{lem:abd:swsr:conducive},
	except for the proof in Case~B, where $\ox_1$ and $\ox_2$ are both reads.
	Here we appeal to the special way in which $S$ is constructed for the multi-reader ABD algorithm
	based on the $\dashrightarrow_{H'}$ relation, whereby $\ox_1 \dashrightarrow_{H'} \ox_2$ implies $\ox_1 \rightarrow_S \ox_2$
		for weakly connected operation execution pairs.
\end{proof}

\begin{theorem}\label{thm:abd:swmr:rlin}
	Let $H = (E, <_E)$ be a history of multi-reader ABD for $N=3$ processors, and let $H'$ be its completion.
	Fix any total order $<_T$ that refines $<_E$, and let $S$ be the corresponding R1-linearization.
	Then $H$ is R3-linearizable, and $S$ its R3-linearization.
\end{theorem}
\begin{proof}
	The completion $H'$ is connected by \Cref{lem:abd:connected},
	and the R1-linearization $S$ is R2-conducive by \Cref{lem:abd:swmr:conducive}.
	R3-linearizability with respect to $S$ follows by \Cref{thm:r1r2r3}.
\end{proof}

\section{Quorum-Replicated Key-Value Stores} \label{sec:kv}

Gilbert and Golab's conjecture covers key-value storage systems that use majority quorums for replication,
	and they cite as an example Amazon’s Dynamo key-value store \cite{dynamo:sosp}, which has a number
	of widely-used derivatives including Apache Cassandra \cite{cassandra}.
Such systems store collections of key-value pairs, and provide get and put operations to read and write the value of a given key.
Each key-value pair can be modelled as a shared read/write register object, with the size of a read or write quorum
	tunable individually for each get or put operation.

To a first approximation, if the key-value protocol is configured with majority quorums then the
	protocol behaves similarly to the ABD simulation described in \Cref{sec:abd},
	though it uses only a single round of communication for both gets and puts.\footnote{The majority quorums must be strict,
	as opposed to ``sloppy'' quorums with hinted hand-off used to maintain full write availability under network partitions.}
The shared object corresponding to each key-value pair therefore satisfies the semantics of Lamport's \emph{regular} register \cite{lam:ipc1}.

A closer look at Dynamo reveals that the question of linearizability is more nuanced.
Get operations sometimes return a collection of conflicting data
	versions, and therefore cannot be mapped directly to read operations on a register object.
Some conflicts are resolved automatically using Dynamo's vector clock (syntactic reconciliation),
	but the version tree occasionally forks into parallel branches that must be merged
	using application-specific business logic (semantic reconciliation), such as a ``last write wins''
	rule over uniquely timestamped versions.
Applications with stateful sessions can further enforce \emph{monotonic reads}, avoiding old-new inversions within a session.
Thus, the exact implementation of a read/write register from Dynamo get/put operations depends
	on application-side session management and reconciliation policies,
	leaving Gilbert and Golab's conjecture somewhat open to interpretation.



We believe that a Dynamo-style key-value store configured with strict majority quorums
	and the ``last write wins'' reconciliation policy can yield R3-linearizable histories in special cases.
Assume a single reader maintaining a continuous session, which ensures monotonicity across all reads,
	and a single writer who assigns data timestamps using a local clock, which avoids clock synchronization issues.
%
As in our analysis of ABD in \Cref{sec:abd}, majority quorums guarantee connected histories, and
	under the single-writer single-reader assumption the histories have R2-conducive R1-linearizations,
	allowing the application of \Cref{thm:r1r2r3}.

Modern implementations of the Dynamo design also provide mechanisms for strong consistency beyond session guarantees.
Notably, recent versions of Apache Cassandra \cite{cassandra} implement lightweight transactions
	using Paxos consensus \cite{Lamport:paxos}, yielding classically linearizable behaviour for multiple writers and readers;
	by arguments similar to our Raft analysis in \Cref{sec:raft}, we believe such systems are R3-linearizable.
\section{Conclusion} \label{sec:conc}
Our analysis of relativistic linearizability in this paper validates Gilbert and Golab's proof technique \cite{gilgol},
	and partly confirms their conjecture with respect to state machine replication, the ABD shared register construction,
	and quorum-replicated key-value stores.
We establish R3-linearizability for all three algorithms, albeit only in special cases for ABD and key-value stores.
To our knowledge, these are the first rigorous proofs of R3-linearizability for any distributed algorithm,
	and the first applications of Gilbert and Golab's proof technique.
Jayanti's recent result \cite{jayanti} that all classically linearizable \emph{asynchronous} algorithms are R2-linearizable
	complements ours but does not subsume our R3-linearizability proofs, a strictly stronger property.
We leave open a broader analysis of the Gilbert-Golab conjecture across more interpretations and special cases
	(e.g., multi-reader ABD for $N > 3$), as well as a systematic study of relativistic linearizability
	for algorithms that break the asynchronous assumption by using physical clocks for coordination.

\begin{acks}
	We acknowledge the support of the Natural Sciences and Engineering Research Council of Canada (NSERC).
\end{acks}

\vfill
\break

\bibliographystyle{ACM-Reference-Format}
\bibliography{refs}

\end{document}